\documentclass[twoside]{article}

\usepackage[accepted]{aistats2023}
\usepackage[dvipsnames]{xcolor}
\usepackage{mathtools}
\usepackage{amssymb,amsfonts,amsmath,amsthm} %ams
\usepackage{enumerate,enumitem,tikz,graphicx,mathrsfs,eucal,verbatim, bbm, derivative}
\usetikzlibrary{calc}
\usepackage[hidelinks]{hyperref}
\usepackage{caption}
\usepackage{subcaption}
\usepackage{wrapfig}
\usepackage{graphbox} 
\usepackage{comment}
\usepackage{algorithm}
\usepackage{algpseudocode}
\usepackage{nicefrac}
\usepackage{array}

\usepackage{makecell}
\usepackage[thinlines]{easytable}
\usepackage{booktabs}

\usepackage[strict]{changepage}
\usepackage{manfnt}
\usepackage{multicol}
 \usepackage[
    backend=biber,
    style=authoryear,
    maxcitenames=1,
    uniquelist=false
  ]{biblatex}
\theoremstyle{plain}% default
\newtheorem{thm}{Theorem}[section]

\newtheorem{prop}[thm]{Proposition}

\theoremstyle{definition}
\newtheorem{defn}{Definition}
\theoremstyle{remark}

\DeclareMathOperator{\vol}{\textnormal{Vol}}

\begin{document}


\runningauthor{G. L. Marchetti, V. Polianskii, A. Varava, F. T. Pokorny, D. Kragic}

\twocolumn[

\aistatstitle{An Efficient and Continuous Voronoi Density Estimator}

\aistatsauthor{Giovanni Luca Marchetti \And Vladislav Polianskii \And  Anastasiia Varava}

\aistatsauthor{Florian T. Pokorny \And Danica Kragic}

\aistatsaddress{ School of Electrical Engineering and Computer Science, KTH Royal Institute of Technology\\
    Stockholm, Sweden  } ]

\begin{abstract}
We introduce a non-parametric density estimator deemed Radial Voronoi Density Estimator (RVDE). RVDE is grounded in the geometry of Voronoi tessellations and as such benefits from local geometric adaptiveness and broad convergence properties. Due to its radial definition RVDE is continuous and computable in linear time with respect to the dataset size. This amends for the main shortcomings of previously studied VDEs, which are highly discontinuous and computationally expensive. We provide a theoretical study of the modes of RVDE as well as an empirical investigation of its performance on high-dimensional data. Results show that RVDE outperforms other non-parametric density estimators, including recently introduced VDEs. 
\end{abstract}

\section{INTRODUCTION}\label{sec:intro}
The problem of estimating a Probability Density Function (PDF) from a finite set of samples lies at the heart of statistics and arises in several practical scenarios (\cite{pointpatterns, scottdensity}). Among density estimators, the non-parametric ones aim to infer a PDF through a closed formula. Differently from parametric methods, they do not require optimization and ideally provide an estimated PDF which is simple, interpretable and computationally efficient. Two traditional examples of non-parametric density estimators are the Kernel Density Estimator (KDE; \cite{kdebook}, \cite{rosenblattkde}) and histograms (\cite{histograms}, \cite{pearsonhistograms}). KDE consists of a mixture of local copies of a kernel around each datapoint while histograms partition the ambient space into local cells (`bins') where the estimated PDF is constant. 
\begin{center}
\begin{figure}[t]
 \centering
 \includegraphics[width=.95\linewidth]{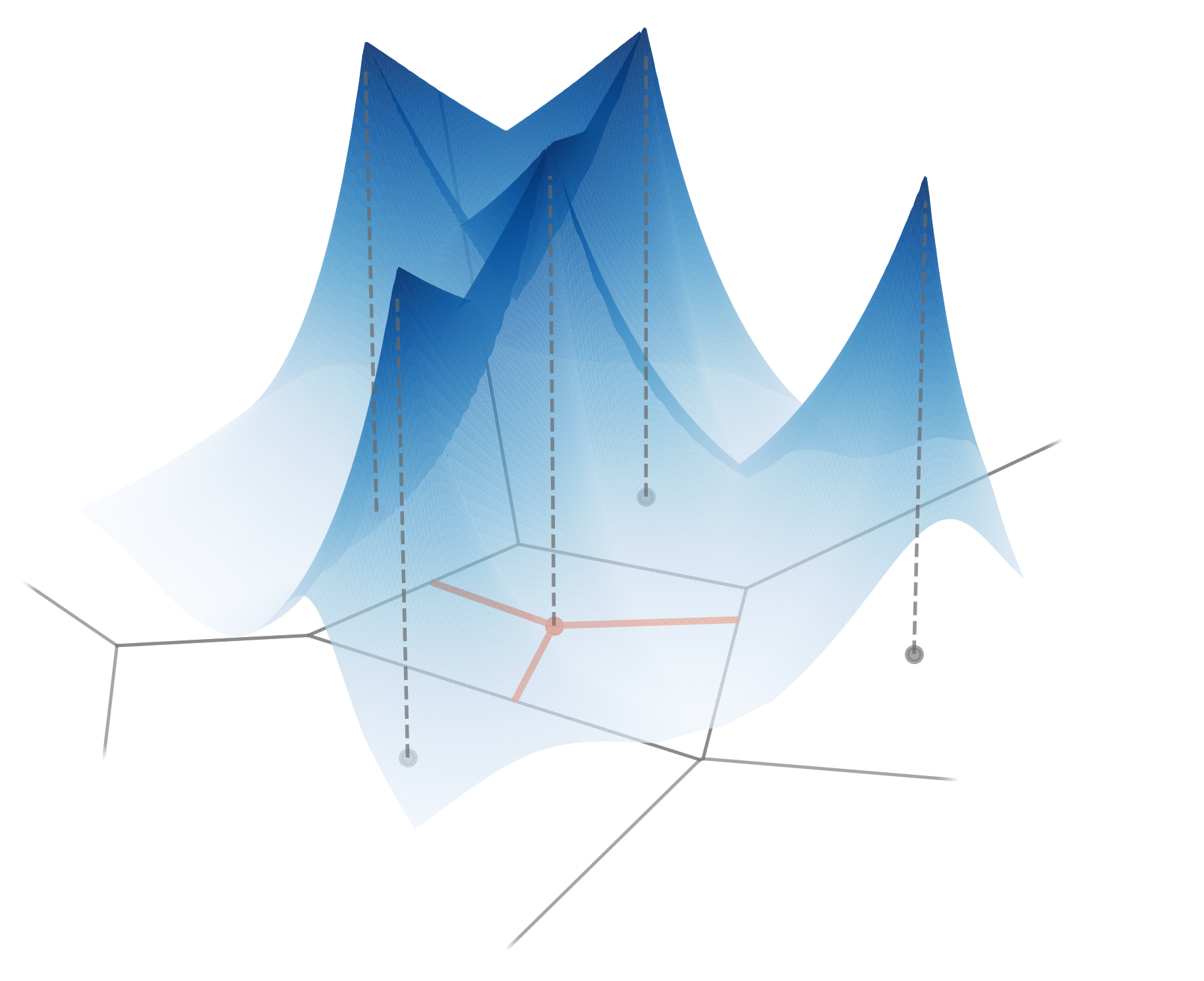} 
\caption{An example of a density estimated via RVDE. The Voronoi tessellation is depicted in solid gray. The estimated density is defined by the property that its conical integral over the rays originating from the datapoints (orange) is constant.}
\label{firstpage}
\end{figure}
\end{center}
Both histograms and KDE suffer from bias due to the prior choice of a local geometric structure i.e., the bins and the kernel respectively. This bias gets exacerbated in high-dimensional ambient spaces. The reason is that datasets grow exponentially in terms of geometric complexity, making a fixed simple geometry unsuitable for estimating high-dimensional densities. This has led to the introduction of the \emph{Voronoi Density Estimator} (VDE; \cite{ord}). VDE relies on the geometric adaptiveness of Voronoi cells, which are convex polytopes defined locally by the data (\cite{voronoibook}). The PDF estimated by VDE is constant on such cells, thus behaving as an adaptive version of histograms. Due to its local geometric properties, VDE possesses convergence guarantees to the ground-truth PDF which are more general than the ones of KDE. 

The geometric benefits of VDE, however, come with a number of shortcomings. First, the Voronoi cells and in particular their volumes are computationally expensive to compute in high dimensions. Although this has been recently attenuated by proposing Monte Carlo approximations (\cite{pol22vor}), VDE falls behind methods such as KDE in terms of computational complexity. Second, VDE (together with its generalized version from \cite{pol22vor} deemed CVDE) is highly discontinuous on the boundaries of Voronoi cells. The estimated PDF consequently suffers from large variance and instability with respect to the dataset. This is again in contrast to KDE, which is continuous in its ambient space.  

In this work, we propose a novel non-parametric density estimator deemed \emph{Radial Voronoi Density Estimator} (RVDE) which addresses the above challenges. Similarly to VDE, RVDE integrates to a constant on Voronoi cells and thus shares its local geometric advantages and convergence properties. In contrast to VDE, RVDE is continuous and computable in linear time with respect to the dataset size. The central idea behind RVDE is to define the PDF radially from the datapoints so that the (conic) integral over the ray cast in the corresponding Voronoi cell is constant (see Figure \ref{firstpage}). This is achieved via a `radial bandwidth' which is defined implicitly by an integral equation. Intuitively, the radial approach reduces the high-dimensional geometric challenge of defining a Voronoi-based estimator to a one-dimensional problem. This avoids the expensive volume computations of the original VDE and guarantees continuity because of the fundamental properties of Voronoi tessellations. Another important aspect of RVDE is its geometric distribution of modes. We show that the modes either coincide with the datapoints or lie along the edges of the Gabriel graph (\cite{gabriel1969new}) depending on a hyperparameter analogous to the bandwidth in KDE. 

We compare RVDE with CVDE, KDE and the adaptive version of the latter in a series of experiments. RVDE outperforms the baselines in terms of the quality of the estimated density on a variety of datasets. Moreover, it runs significantly faster and with lower sampling variance compared to CVDE. This empirically confirms that the geometric and continuity properties of RVDE translate into benefits for the estimated density in a computationally efficient manner. We provide an implementation of RVDE (together with baselines and experiments) in $C\texttt{++}$ at a publicly available repository \footnote{\url{https://github.com/giovanni-marchetti/rvde}}. The code is parallelized via the OpenCL framework and comes with a Python interface. In summary our contributions include: 
\begin{itemize}
    \item A novel density estimator (RVDE) based on the geometry of Voronoi tessellations which is continuous and computationally efficient. 
    \item A complete study of the modes of RVDE and their geometric distribution. 
    \item An empirical investigation comparing RVDE to KDE (together with its adaptive version) and previously studied VDEs. 
\end{itemize}

\section{RELATED WORK}\label{relwork}

\subsection{Non-parametric Density Estimation} 

Non-parametric methods for density estimation trace back to the introduction of histograms (\cite{pearsonhistograms}). Histograms have been extended by considering bin geometries beyond the canonical rectangular one, for example triangular (\cite{scottrecrangles}) and hexagonal (\cite{hexagons}) geometries. Another popular density estimator is KDE, first discussed by \cite{rosenblattkde} and \cite{parzenkde}. The estimated density is a mixture of copies of a priorly chosen distribution (`kernel') centered at the datapoints. KDE has been extended to the multivariate case (\cite{izemanmulti, silvermanmulti}) and has seen improvements such as bandwidth selection methods (\cite{band1, band2}) and algorithms for adaptive bandwidths (\cite{wang2007bandwidth, van2017variable}). Applications of KDE include estimation of traffic incidents (\cite{kdetraffic}), of archaeological data (\cite{kdearcheo}) and of wind speed (\cite{kdewind}) to name a few. As discussed in Section \ref{sec:intro}, both KDE and histograms suffer from lack of geometric adaptiveness due to the choice of prior local geometries.

Another class of non-parametric methods are the orthogonal density estimators (\cite{vannucci1995nonparametric, masry1997multivariate}). Those consist of choosing a discretized orthonormal basis of functions and computing the coefficients of the ground-truth density via Monte-Carlo integration over the dataset. When the basis is the Fourier one, the estimator is referred to as `wavelet estimator'. The core drawback is that orthogonal density estimators do not scale efficiently to higher dimensions. When considering canonical tensor product bases the complexity grows exponentially w.r.t. the dimensionality (\cite{walter1995estimation}), making the estimator unfeasible to compute.    

\subsection{Voronoi Density Estimators}
The first Voronoi Density Estimator (VDE) has been pioneered by \cite{ord}. The estimated density relies on Voronoi tessellations in order to achieve local geometric adaptiveness. This is the main advantage of VDE over methods such as KDE. The original VDE has seen applications to real-world densities such as neurons in the brain (\cite{voronoineuronal}), photons (\cite{voronoiphotons}) and stars in a galaxy (\cite{voronoiastronomy}). However, the method is not immediately extendable to high-dimensional spaces because of unfeasible computational complexity of volumes and abundance of unbounded Voronoi cells. This has been only recently amended by \cite{pol22vor} by introducing approximate numerical algorithms and by shaping of the density via a kernel. In the present work, we aim to design an alternative version of the original VDE which is continuous and does not rely on volume computations. The resulting estimator is thus stable and computationally efficient while still benefiting from the geometry of Voronoi tessellations.

\section{BACKGROUND}\label{sec:vde}
In this section we recall the class of non-parametric density estimators which we will be interested in throughout the present work. To this end, let $P \subseteq \mathbb{R}^n$ be a finite set and consider the following  central notion from computational geometry.

\begin{defn}
The \emph{Voronoi cell}\footnote{Sometimes referred to as \emph{Dirichlet cell}.} of $p \in P$ is: 
\begin{equation}
C(p) = \{ x\in \mathbb{R}^n \ | \ \forall q \in P \  d(x,q) \geq d(x,p) \}.
\end{equation}
\end{defn}
The Voronoi cells are convex polytopes that intersect at the boundary and cover the ambient space $\mathbb{R}^n$. The collection $\{C(p) \}_{p \in P}$ is referred to as  \emph{Voronoi tessellation} generated by $P$. Note that although the Voronoi tessellations are defined in an arbitrary metric space, the resulting cells might be non-convex for distances different from the Euclidean one. Since convexity will be crucial for the following constructions, we stick to the Euclidean metric for the rest of the work.  

We call \emph{density estimator} any mapping associating a probability density function $f_P \in L^1(\mathbb{R}^n)$ to a finite set $P \subseteq \mathbb{R}^n$. The following class of density estimators generalizes the original one by \cite{ord}. 
\begin{defn}
 A \emph{Voronoi Density Estimator} (VDE) is a density estimator $P \mapsto f_P$ such that for each $p\in P$: 
 \begin{equation}\label{generalvde}
 \int_{C(p)} f_P(x) \  \textnormal{d}x = \frac{1}{|P|}. 
 \end{equation}
 \end{defn}
  VDEs stand out among density estimators for their geometric properties. This is because the Voronoi cells are arbitrary polytopes that are adapted to the local geometry of data. For VDEs all the Voronoi cells have the same estimated probability, making such estimators locally adaptive from a geometric perspective. This is reflected, for example, by the general convergence properties of VDEs. The following is the main theoretical result from \cite{pol22vor}.

  \begin{thm}\label{convergence}
Let $P \mapsto f_P$ be a VDE and suppose that $P$ is sampled from a probability density $\rho \in L^1(\mathbb{R}^n)$ with support in the whole $\mathbb{R}^n$. For $P$ of cardinality $m$ consider the probability measure $\mathbb{P}_m = f_P \textnormal{d}x$ which is random in $P$.  Then the sequence $\mathbb{P}_m$ converges to $\mathbb{P} = \rho \textnormal{d}x$ in distribution w.r.t. $x$ and in probability w.r.t. $P$.  Namely, for any measurable set $E \subseteq \mathbb{R}^n$ the sequence of random variables $\mathbb{P}_m(E)$ converges in probability to the constant $\mathbb{P}(E)$.
\end{thm}

In contrast, the convergence of other density estimators such as KDE requires the kernel bandwidth to vanish asymptotically (\cite{devroye1979l1}). The bandwidth vanishing is necessary in order to amend for the local geometric bias inherent in KDE as discussed in Section \ref{sec:intro}. 

The following canonical construction of a VDE deemed Compactified Voronoi Density Estimator (CVDE) is discussed by \cite{pol22vor}. Given an integrable kernel $K: \ \mathbb{R}^n \times \mathbb{R}^n \rightarrow \mathbb{R}_{>0}$ the estimated density is defined as 
\begin{equation}
 f_P(x) =   \frac{K(p,x)}{|P| \vol_p(C(p))}
\end{equation}
where $p$ is the closest point in $P$ to $x$ and $\vol_p(C(p))=  \int_{C(p)} K(p,y) \ \textnormal{d}y$. The latter volumes are approximated via Monte Carlo methods since they become unfeasible to compute exactly as dimensions grow. The resulting density inherits the same regularity as $K$ when restricted to each Voronoi cell but jumps  \emph{discontinuously} when crossing the boundary of Voronoi cells (see Figure \ref{fig:heats}). Motivated by this, the goal of the present work is to introduce a continuous and efficient VDE.

\section{METHOD}
 \subsection{Radial Voronoi Density Estimator}\label{sec:contvdedef}
In this section we outline a general way of constructing a VDE with continuous density function. Our central idea is to define the latter \emph{radially} w.r.t. the datapoints $p \in P$. We start by rephrasing the integral over a Voronoi cell (Equation \ref{generalvde}) in spherical coordinates:
\begin{align}\label{sphericalint}
\begin{split}
   & \int_{C(p)} f_P(x) \ \textnormal{d}x = \\
   =& \int_{\mathbb{S}^{n-1}} \underbrace{\int_{0}^{l(p + \sigma)} t^{n-1} f_P(p + t\sigma) \ \textnormal{d} t }_{\textnormal{Conical Integral}}\textnormal{d} \sigma.
   \end{split}
\end{align}
Here $\mathbb{S}^{n-1} \subseteq \mathbb{R}^n$ denotes the unit sphere and $l(x) \in [0, +\infty]$ denotes the length of the segment contained in $C(p)$ of the ray cast from $p$ passing through $x$ i.e., 
\begin{equation}\
    l(x) = \sup \left\{ t \geq 0 \ | \ p + t\frac{x-p}{d(x,p)} \in C(p) \right\}.
\end{equation}  
\begin{center}
\begin{figure}
        \includegraphics[width=.65\linewidth]{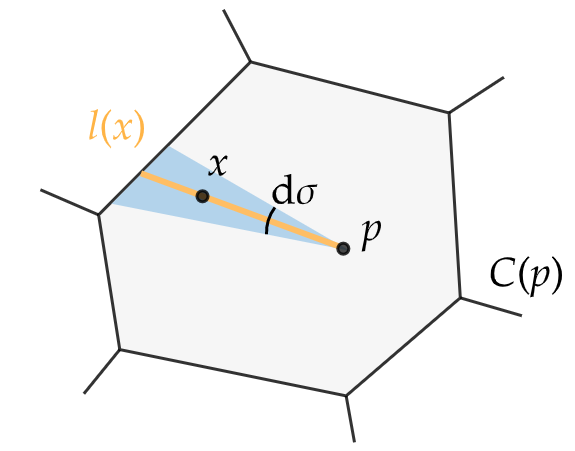} 
        \centering
\caption{Illustration of the quantity $l(x)$ involved in the definition of RVDE (Definition \ref{def:rvde}). The estimated density integrates to a constant over all the infinitesimal cones (blue) in $C(p)$ originating from $p$. }\label{figradial}
\end{figure}
\end{center}
We refer to Figure \ref{figradial} for a visual illustration. Note that $l(x)$ is defined for $x \not = p$ and is continuous in its domain since $l(x) =  d(x,p) = d(x,q)$ for $x \in C(p) \cap C(q)$. 

We aim to solve Equation \ref{generalvde} by forcing the conical integral in Equation \ref{sphericalint} to be \emph{constant}. To this end, we fix a continuous and strictly decreasing function $K : \ \mathbb{R}_{>A} \rightarrow \mathbb{R}_{\geq 0}$ (a `kernel') defined on a half-line $\mathbb{R}_{> A}$, $A < 0$, with the property that $t^{n-1} K(t)$ is integrable on $\mathbb{R}_{>0}$. By an ansatz we look for a density in the form 
\begin{equation}\label{contvde}
f_P(x) = \frac{K\left(\beta(l(x)) d(x,p)\right)}{ \alpha |P| \textnormal{Vol}(\mathbb{S}^{n-1})}
\end{equation}
where $\alpha > 0$ is a hyperparameter and $\beta : \ \mathbb{R}_{>0} \rightarrow \mathbb{R}$ is a function that we would like to determine. The latter intuitively represents a \emph{radial bandwidth}. The density $f_P$ is continuous since the discontinuity of $l$ at $x =p$ is amended by the vanishing of $d(x,p)$. Equation \ref{generalvde} is satisfied if for every $l >0$:
\begin{equation}\label{eqbeta}
    \int_0^l t^{n-1}K(\beta(l)t) \ \textnormal{d}t = \alpha.
 \end{equation}
Since $K$ is strictly decreasing, the above expression always has a unique solution $\beta(l) > \frac{A}{l}$ assuming that $t^{n-1}K(t)$ is not integrable around $A$. Such a guaranteed solution can be computed via any root-finding algorithm and is continuous w.r.t. $l$. We provide an analysis of the function $\beta$ and a discussion of the Newton-Raphson method for its computation in Section \ref{sec:study}. 

The derivations above bring us to the following definition. 
\begin{defn}\label{def:rvde}
Fix an $\alpha > 0$ and a continuous function $K : \ \mathbb{R}_{>A} \rightarrow \mathbb{R}_{> 0}$ with the domain bound $A < 0$. Assume the following: 
\begin{itemize}
    \item $K(0) = 1$,
    \item $K$ is strictly decreasing,
    \item $|t|^{n-1}K(t)$ is integrable around $+\infty$ but not integrable around $A$.
\end{itemize}
The \emph{Radial Voronoi Density Estimator} (RVDE) is the density estimator defined by Equation \ref{contvde} where $\beta$ is the function defined implicitly by Equation \ref{eqbeta}.
\end{defn}

The following two standard families of kernels $K$ satisfy the above requirements: 
% \begin{equation}
%   K(t) = e^{-\lambda t}
% \hspace{2cm}
% K(t) = \frac{1}{(t+1)^k}
% \end{equation}

\begin{equation}\label{eq:kernels}
\begin{tabular}{ccc}
\emph{Exponential} &  \hspace{.7cm}  & \emph{Rational}  \\
 $K(t) = e^{-t}$ & \hspace{.2cm} & $K(t) = \frac{1}{(t+1)^k}$
\end{tabular}
\end{equation}

where $k >  n$. The domain bounds are $A=-\infty$ and $A = -1$ respectively. When $n=1$ and $K$ is the exponential kernel, the function $\beta$ is closely related to the Lambert $W$ function (\cite{corless1996lambertw}) via the expression:
\begin{equation}
\beta(l) = \frac{1}{\alpha} + W\left( - \frac{l}{\alpha} e^{-\frac{l}{\alpha}}  \right).
\end{equation}
We provide an empirical comparison between the two kernels from Equation \ref{eq:kernels} in Section \ref{sec:kernelcomp}. 
    
The intuition behind the hyperparameter $\alpha$ is that it controls the trade-off between the amount of density concentrated around $P$ and away
from it (i.e., on the boundary of Voronoi cells). Indeed as $\alpha \to 0^+$ RVDE tends (in distribution) to the discrete empirical measure over $P$ while as $\alpha \to +\infty$ it tends to a measure concentrated on the boundary of Voronoi cells. This can be deduced from Equation \ref{eqbeta} since $\beta(l)$ tends to $+\infty$ and to $A/l$ respectively and thus Equation \ref{contvde} tends to $0$ for $d(x,p) \not = 0, l(x)$. This intuition around $\alpha$ will be corroborated by Proposition \ref{modesprop}, where we study how it controls the distribution of modes of RVDE and consequently propose a heuristic selection procedure. 

\begin{figure*}[th!]
\captionsetup[subfigure]{justification=centering}
    \centering
    \hspace{1em}
    \begin{subfigure}[b]{.23\linewidth}
        \centering
        \includegraphics[width=\linewidth]{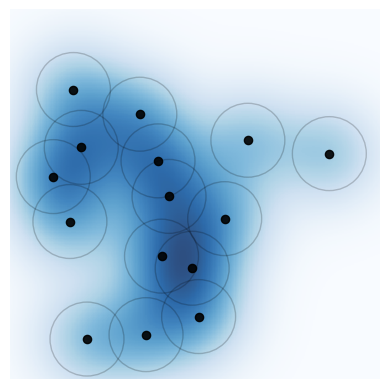}
        \subcaption*{KDE \\ (\cite{rosenblattkde})}
    \end{subfigure}
    \hfill
    \begin{subfigure}[b]{.23\linewidth}
        \centering
        \includegraphics[width=\linewidth]{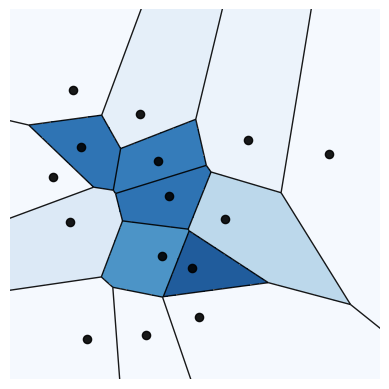}
        \subcaption*{Original VDE \\ (\cite{ord})}
    \end{subfigure}
    \hfill
    \begin{subfigure}[b]{.23\linewidth}
        \centering
 \includegraphics[width=\linewidth]{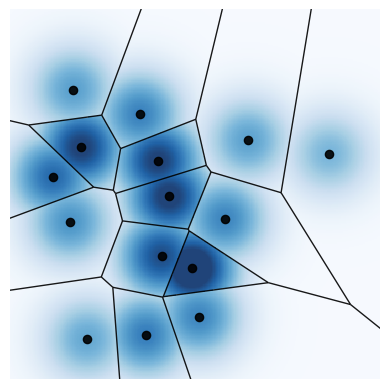}
        \subcaption*{CVDE \\ (\cite{pol22vor})}
    \end{subfigure}
    \hfill
    \begin{subfigure}[b]{.23\linewidth}
        \centering
        \includegraphics[width=\linewidth]{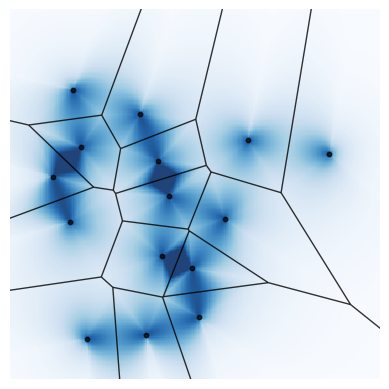}
        \subcaption*{RVDE \\ (Ours)}
    \end{subfigure}
    \hspace{1em}

    \caption{From left to right: heatmaps of KDE, of the two VDEs from the literature and of our RVDE.}
    \label{fig:heats}
\end{figure*}

\subsection{Computational Complexity and Sampling}\label{sec:compute}
We now discuss the computational cost of evaluating RVDE at a point $x \in \mathbb{R}^n$. To begin with, the closest $p \in P$ to $x$ can be found in logarithmic time w.r.t. $|P|$ by organizing $P$ in an efficient data structure for nearest neighbor lookups such as a $k$-$d$ tree. Then $l(x)$ can be computed in linear time via the following closed expression (\cite{pol22vor}): 
\begin{equation}\label{eq:lqz}
    l(x) = \min_{q \not = p, \ l^q(x) \geq 0} l^q(x)
\end{equation}
where 
\begin{equation}
     l^q(x) = \frac{ d(q,p)^2 }{2 \left\langle \frac{x-p}{d(x,p)}, q-p \right\rangle }.
\end{equation}
The computational cost of evaluating $f_P(x)$ is thus linear w.r.t. $|P|$. The remaining compute essentially reduces to solving Equation \ref{eqbeta}, which depends on the integrator, the root-finder algorithm adopted and the desired precision. 

The formulation of RVDE enables a simple and efficient procedure for sampling from the estimated density. In order to sample, one first chooses a $p \in P$ uniformly since $f_P$ integrates to $\frac{1}{|P|}$ on each Voronoi cell (Equation \ref{generalvde}). Since $t^{n-1}f_P(p + t\sigma)$ integrates to a constant on the ray $r = \{ p + t\sigma \}_{t\geq 0} \cap C(p)$ for every $\sigma \in \mathbb{S}^{n-1}$, one then samples $\sigma$ uniformly from the sphere. Finally one samples $t$ from the one-dimensional density $t^{n-1}K(t)$ restricted to the interval $[0, l(p+\sigma)]$. The computational complexity of the latter step depends on the kernel as well as of the sampling method. The result of the sampling is $p+t\sigma$. Because of the computational cost of $l(p + \sigma)$, the sampling complexity of RVDE is linear w.r.t. $|P|$.

% \begin{center}
% \begin{figure}
%     \centering
%         \includegraphics[width=.5\linewidth]{images/contvde_heatmap.png} 
% \caption{Heatmap of RVDE on a simple dataset.}
% \end{figure}
% \end{center}

RVDE is more efficient than the VDE discussed by \cite{pol22vor} (see the end of Section \ref{sec:vde}). The latter relies on Monte Carlo integration for numerical approximation of volumes of Voronoi cells and has complexity $O(\Sigma|P|^2)$ where $\Sigma$ is the number of Monte Carlo samples. Compared to KDE, RVDE has the same computational complexity (for both evaluation and sampling) while retaining the geometric benefits of a VDE. 

\subsection{Study of $\beta$ and Modes}\label{sec:study}
In this section we discuss qualitative properties and computational aspects of the function $\beta$ defined implicitly by Equation \ref{eqbeta} and consequently characterize the modes of RVDE. We start by presenting an explicit expression of the Newton-Raphson iteration for the computation of $\beta(l)$. 

\begin{prop}\label{newtonmeth}
Fix $l>0$ and suppose $K \in C^1(\mathbb{R}_{>A})$ i.e., it is continuously differentiable. Then the iteration $\beta_{m+1}$ of the Newton-Raphson method for computing $\beta(l)$ by solving Equation \ref{eqbeta} takes form: 
% \begin{equation}
% \begin{split}
%     &\beta_{m+1} =  \beta_m \cdot \\
%     & \Biggl( 1 + \frac{1}{n} \Biggl(1 - \frac{l^nK(\beta_m l) - n\alpha}{l^nK(\beta_m l)  - n \int_0^l t^{n-1}K(\beta_mt) \ \textnormal{d}t}  \Biggl) \Biggl). 
% \end{split}
% \end{equation}
\begin{flalign}
\beta_{m+1} =  &&
\end{flalign}
\begin{equation*}
=\beta_m + \frac{\beta_m}{n} \Biggl(1 - \frac{l^nK(\beta_m l) - n\alpha}{l^nK(\beta_m l)  - n \int_0^l t^{n-1}K(\beta_mt) \ \textnormal{d}t}  \Biggl) . 
\end{equation*}
Moreover, if $K$ is convex then the Newton-Raphson method converges for any initial value $\beta_0$ i.e., $\lim_{m \to \infty}\beta_m = \beta(l)$. 
\end{prop}
We refer to the Appendix for a proof. Note that the convexity assumption is satisfied by both the kernels from Equation \ref{eq:kernels}. Proposition \ref{newtonmeth} enables to compute $\beta(l)$ and together with Section \ref{sec:compute} provides all the algorithmic details for implementing RVDE. 

Next, we outline a qualitative study of the function $l \mapsto \beta(l)$. 
\begin{prop}\label{diffeq}
The function $\beta: \ \mathbb{R}_{>0} \rightarrow \mathbb{R}$ is increasing, has a zero at $l= (n \alpha)^{\frac{1}{n}}$ and has an horizontal asymptote: 
\begin{equation}
\lim_{l \to + \infty} \beta(l) = \left( \frac{1}{\alpha} \int_0^\infty t^{n-1}K(t) \ \textnormal{d}t \right)^{\frac{1}{n}}. 
\end{equation}
Moreover if $K \in C^1(\mathbb{R}_{>A})$ then $\beta \in C^1(\mathbb{R}_{>0})$ and it satisfies the differential equation:
\begin{equation}
 \left( l- \frac{n \alpha}{l^{n-1}K(\beta(l)l)} \right) \frac{\textnormal{d}\beta}{\textnormal{d}l}(l) = - \beta(l). 
\end{equation}
\end{prop}

We refer to the Appendix for a proof. As discussed in Section \ref{sec:contvdedef}, $\beta$ generalizes the Lambert $W$ function. The properties and the differential equation from Proposition \ref{diffeq} generalize their well-known instances for the $W$ function (\cite{corless1996lambertw}).

We now focus on the study of modes. Our goal is to describe the modes of RVDE completely. This is an advantage over density estimators such as KDE, where the modes are challenging to describe and to compute approximately (\cite{lee2021finding, comaniciu2002mean}). Denote by $\varepsilon = (n \alpha)^{\frac{1}{n}}$ the zero of $\beta$. Proposition \ref{diffeq} implies that for $x\in \mathbb{R}^n$, the density $f_P$ decreases radially w.r.t. $p$ in the direction of $x$ if $l(x) >  \varepsilon$ and increases otherwise. This leads to the following result. 
\begin{prop}\label{modesprop}
    The modes of $f_P$ are classified as follows:  
\begin{enumerate}[label=(\arabic*)]
    \item $p \in P$ if $d(p, q)> 2 \varepsilon$ for every Voronoi cell $C(q)$ adjacent to $C(p)$,
    \item $\frac{p + q}{2}$ for $p, q \in P$ if $\frac{p + q}{2} \in C(p) \cap C(q)$ and $d(p,q) < 2 \varepsilon$,
    \item all points belonging to the segment $[p,q]$ for $p, q \in P$ if $\frac{p + q}{2} \in C(p) \cap C(q)$ and $d(p,q) = 2\varepsilon$. 
\end{enumerate}
\end{prop}

We refer to the Appendix for a proof and to Figure \ref{fig:mode} for an illustration. Since $\varepsilon$ depends monotonically on the hyperparameter $\alpha$, the latter controls the threshold for distances between adjacent points in $P$ below which the mode gets pushed away from such points towards the boundary of the Voronoi cells. Intuitively, $\alpha$ determines the extent by which points in $P$ are considered `isolated' (i.e., a mode) or otherwise get `merged' by placing a mode between them.

\begin{figure*}[th!]
    \centering
         \includegraphics[width=.6\linewidth]{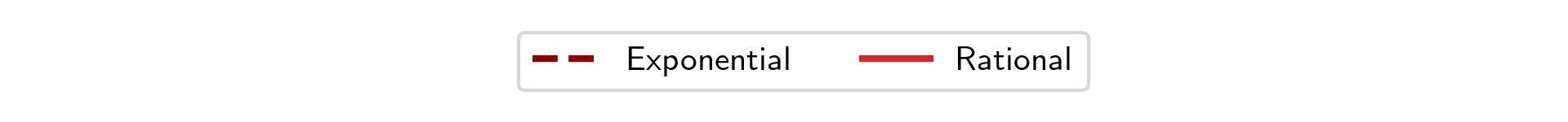}
         
    \begin{subfigure}[b]{.31\linewidth}
        \centering
        \includegraphics[width=\linewidth]{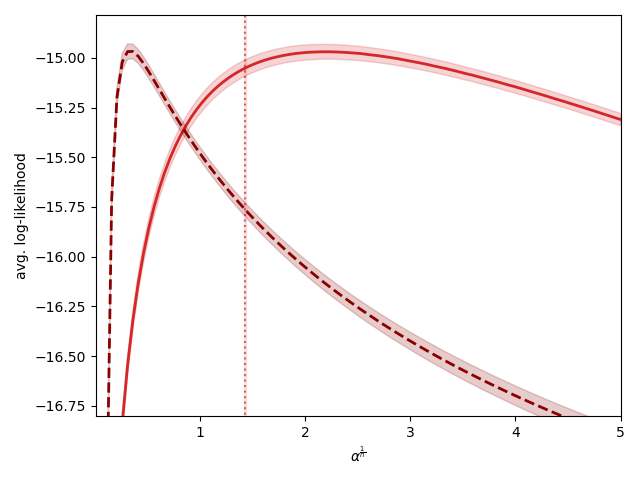}
        \subcaption*{Gaussian}
    \end{subfigure}
    \hspace{1em}
    \begin{subfigure}[b]{.31\linewidth}
        \centering
 \includegraphics[width=\linewidth]{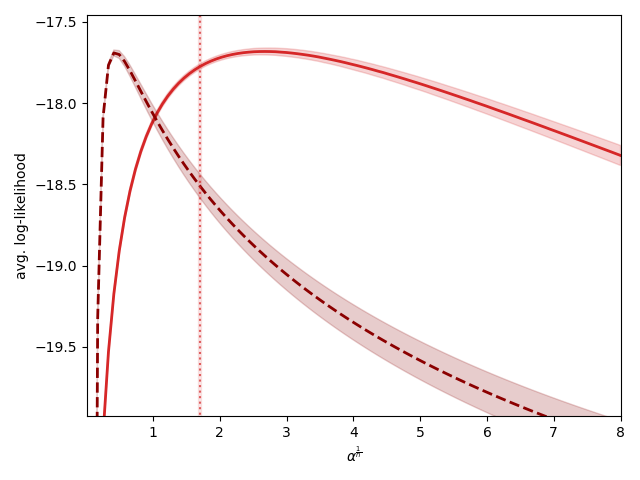}
    \subcaption*{Laplace}
    \end{subfigure}
    \hspace{1em}
    \begin{subfigure}[b]{.31\linewidth}
        \centering
        \includegraphics[width=\linewidth]{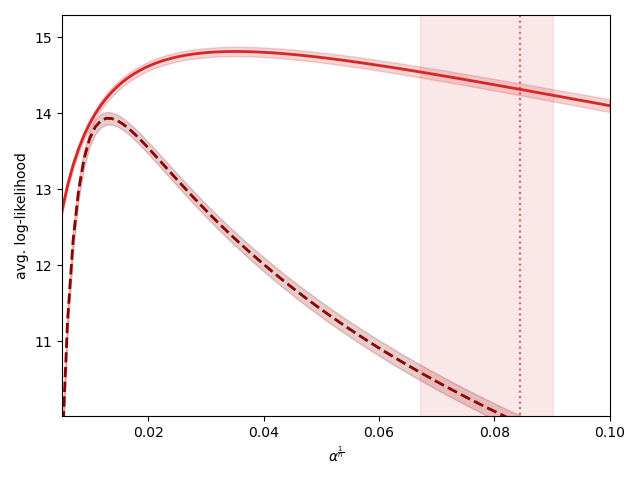}
        \subcaption*{Dirichlet}
    \end{subfigure}
    % \begin{subfigure}[b]{.1\linewidth}
    %     % \centering
    %     % \includegraphics[width=\linewidth]{images/kernel_comparison/kernel_labels.png}
    %     % \vspace{12pt}
    %     % \subcaption*{\hspace{0.01pt}}
    % \end{subfigure}
    \caption{Comparison of the two kernels for RVDE (Equation \ref{eq:kernels}) on three simple distributions in $10$ dimensions.}
    \label{fig:kernels}
\end{figure*}

An alternative geometric formulation of Proposition \ref{modesprop} is the following. Consider the \textit{Gabriel graph} of $P$ (\cite{gabriel1969new}) containing an edge between $p$ and $q$ iff $\frac{p+q}{2} \in C(p) \cap C(q)$ and discard all the edges of length greater than $2 \varepsilon$. The modes of RVDE are then associated with $(1)$ all isolated vertices, $(2)$ midpoints of edges and $(3)$ whole edges of length $2\varepsilon$. Intuitively, the modes of RVDE are distributed geometrically according to the truncated Gabriel graph.

\begin{figure}[H]
\includegraphics[width=.55\linewidth]{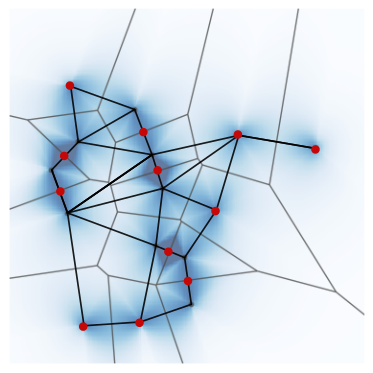} 
\centering
\caption{Illustration of the modes of RVDE (red) together with the Gabriel graph (black).}\label{fig:mode}
\end{figure}

This suggests a possible heuristic procedure for \emph{hyperparameter selection} of $\alpha$. An option is to consider statistics of lengths of edges in the Gabriel graph and choose $2 \varepsilon$ (and thus $\alpha$) as a percentile. The percentile we suggest is $\frac{|P| -1}{|E|}$ where $E$ denotes the set of edges of the Gabriel graph. The intuition is that we wish to avoid modes distributed in cycles. The number of cycles in the Gabriel graph is $|P| - |E| + 1$, from which our suggested percentile follows. This procedure enables to select $\alpha$ automatically and we evaluate it empirically in Section \ref{sec:exper}. However, it comes with a number of limitations. First, the computational complexity of such a procedure is $O(|P|^3)$ because of the construction of the Gabriel graph, which is feasible but might become expensive for large datasets. Another limitation is its independence from the kernel $K$. The selection of $\alpha$ might be satisfying for some kernels but not for others. In our empirical evaluation from Section \ref{sec:exper} we show that for the rational kernel the selected $\alpha$ is close to the optimal one in practice, while for the exponential kernel the selection is further from optimality.

\section{EXPERIMENTS}\label{sec:exper}
Our empirical investigation is organized as follows. First we study RVDE on its own by comparing the different choices of the kernel. We then compare RVDE with other non-parametric density estimators on a variety of datasets. 

\subsection{Evaluation Metrics and Baselines}\label{sec:evmetric}
We evaluate all the density estimators $f_P$ via average log-likelihood on a test set $P_{\textnormal{test}}$ i.e.,
\begin{equation}\label{eqloglik}
\frac{1}{|P_{\textnormal{test}}|}\sum_{x \in P_\textnormal{test}} \log f_P(x).
\end{equation}
This measures whether the estimator assigns high density values to points outside of $P$ but sampled from the same distribution. In order to empirically evaluate the computational complexity, we additionally include runtimes for all the experiments. Our implementations of all the considered density estimators share the same programming framework and are parallelized to a similar degree, making the raw runtimes a fair comparison. We perform experiments on a machine with an AMD Ryzen 9 5950X 16-core CPU and a GeForce RTX 3090 GPU.

We deploy the following non-parametric density estimators as baselines in the experiments. 

\textbf{Kernel Density Estimator} (KDE): given a (normalized)  kernel $K: \mathbb{R}^n \rightarrow \mathbb{R}_{\geq 0}$ the density is estimated as:
\begin{equation}
f_P(x) = \frac{1}{|P| h^n} \sum_{p \in P} K\left( \frac{x-p}{h} \right)
\end{equation}
where $h$ is the bandwidth hyperparameter. 
        
\textbf{Adaptive Kernel Density Estimator} (AdaKDE; \cite{wang2007bandwidth}): a version of KDE where the bandwidth $h_p$ depends on $p \in P$ and is smaller when data is denser around $p$. Specifically, if $f_P(p)$ denotes the standard KDE estimate with a global bandwidth $h$ then $h_p = h \lambda_p$, where:
% $\lambda_p = (g / f_P(p))^{\frac{1}{2}}$ and $g = \prod_{q \in P}f_P(q)^{\frac{1}{|P|}}$.
\begin{equation}\lambda_p = (g / f_P(p))^{\frac{1}{2}}, \quad \quad g = \prod_{q \in P}f_P(q)^{\frac{1}{|P|}}.
\end{equation}

\textbf{Compactified Voronoi Density Estimator} (CVDE; \cite{pol22vor}): the VDE described at the end of Section \ref{sec:vde}. It depends on a kernel $K$ (together with a bandwidth) and is discontinuous on the boundary of Voronoi cells.

% We implement KDE, AdaKDE and CVDE with the standard Gaussian kernel: $K(x) = (2 \pi)^{-\frac{n}{2}} e^{- \frac{1}{2}\| x\|^2}$. 

\subsection{Datasets}\label{secdsets}
In our experiments we consider data of varying nature. This includes both simple synthetic distributions and real-world datasets in high dimensions. For the latter, we consider sound data ($n=21$) and image data ($n=100$). Our datasets are the following. 

{\bf  Synthetic Datasets}: datasets generated from a number of simple densities in $n=10$  dimensions. Both $P$ and $P_\textnormal{test}$ contain $1000$ points in all the cases. The densities we consider are: a standard Gaussian distribution, a standard Laplace distribution, a Dirichlet distribution with parameters $\alpha_i = \frac{1}{n+1}$ and a mixture of two Gaussians with means $\mu_1 = (-0.5, 0 , \cdots, 0)$, $\mu_2 = (0.5, 0 , \cdots, 0) $ and standard deviations $\sigma_1 = 0.1$, $\sigma_2 = 10$ respectively. 
    
{\bf  MNIST} (\cite{deng2012mnist}): a dataset  consisting of $28 \times 28$ grayscale images of handwritten digits which are normalized in order to lie in $[0,1]^{28 \times 28}$. In order to densify the data and obtain more meaningful estimates, we downscale the images to resolution $10 \times 10$. For each experimental run, we sample half of the $60000$ training datapoints in order to evaluate the variance of the estimation. The test set size is $10000$. 
 
{\bf  Anuran Calls} (\cite{Dua:2019}): a dataset consisting of $7195$ calls from $10$ species of frogs which are represented by $21$ normalized mel-frequency cepstral coefficients in $[0,1]^{21}$. We retain $10 \%$ of data for testing and sample half of the training data at each experimental run.

\begin{figure*}[th!]
    \centering
    \includegraphics[width=.6\linewidth]{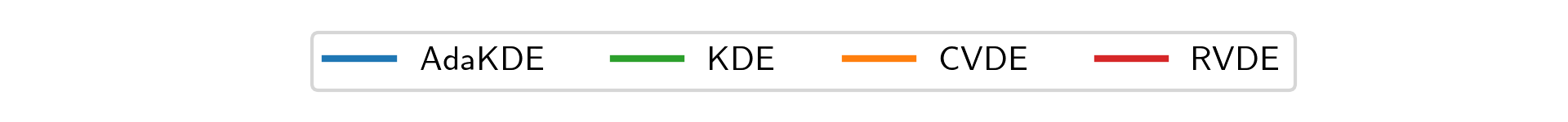}
    
    \begin{subfigure}{.35\linewidth}
    \centering
        \includegraphics[width=\linewidth]{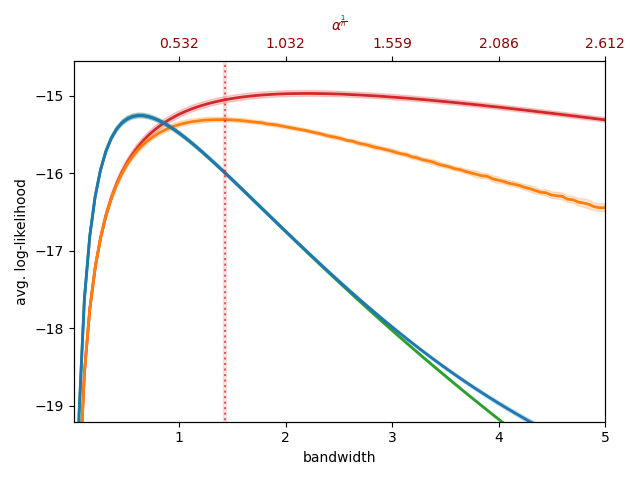}
        \subcaption*{Single Gaussian}
    \end{subfigure}
    \hspace{3em}
    \begin{subfigure}{.35\linewidth}
    \centering
        \includegraphics[width=\linewidth]{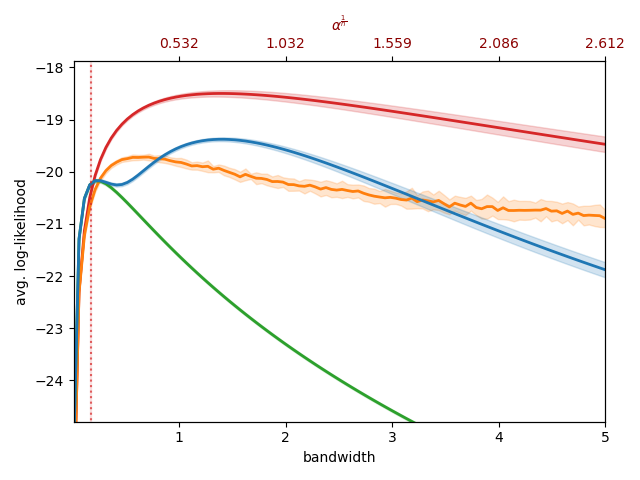}
        \subcaption*{Gaussian Mixture}
    \end{subfigure}
    \begin{subfigure}{.35\linewidth}
    \centering
        \includegraphics[width=\linewidth]{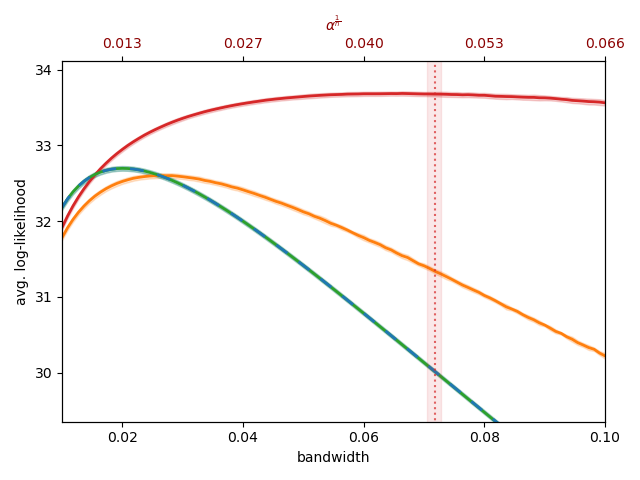}
        \subcaption*{Anuran Calls}
    \end{subfigure}
    \hspace{3em}
    \begin{subfigure}{.35\linewidth}
    \centering
        \includegraphics[width=\linewidth]{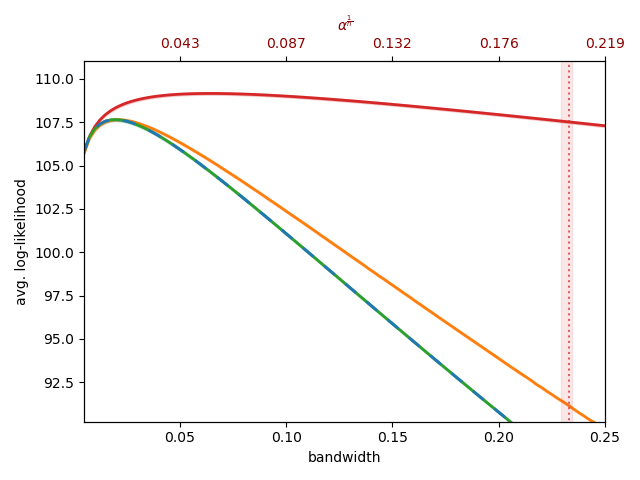}
        \subcaption*{MNIST}
    \end{subfigure}
    % \begin{subfigure}{.1\linewidth}
    % \centering
    %     \includegraphics[width=\linewidth]{images/de_comparison/vskde_labels.png}
    %     \subcaption*{\hspace{1pt}}
    % \end{subfigure}
    \caption{Comparison of the estimators as the bandwidth varies. All the estimators implement the rational kernel.}
    \label{fig:de-comparison}
\end{figure*}

\subsection{Comparison of Kernels}\label{sec:kernelcomp}
Our first experiment consists of a comparison between the rational and the exponential kernel for RVDE (Equation \ref{eq:kernels}) on the synthetic datasets. In what follows the exponent $k$ for the rational kernel is set to $k=n + 1$ for simplicity, where $n$ is the dimension of the ambient space of the dataset considered (in this experiment, $n=10$). 

The results are presented in Figure \ref{fig:kernels}. The plot displays the test log-likelihood (Equation \ref{eqloglik}) as the hyperparameter $\alpha$ varies. The latter is scaled as $\alpha^{\frac{1}{n}}$ in order to be consistent with the visualizations in the following section. The curves on the plot represent mean and standard deviation (shaded areas) over $5$ experimental runs for $100$ bandwidths. The additional vertical lines correspond to the value of the hyperparameter selection heuristic discussed at the end of Section \ref{sec:study}. As can be seen, the performance of the rational kernel is more stable w.r.t. the hyperparameter $\alpha$. The exponential kernel, however, achieves a slightly higher test score with its best hyperparameter on the Gaussian and Laplace datasets. Note that the heuristically chosen $\alpha$ aligns well with the one that achieves the best performance for the rational kernel, but is misaligned for the exponential one. We conclude that the rational kernel is generally a better option unless an extensive hyperparameter search is performed. In what follows we consequently stick to the rational kernel for RVDE.

\subsection{Comparison with Baselines}
In our main experiment we compare the performance of RVDE with the baselines described in Section \ref{sec:evmetric}. We consider the test log-likelihood (Equation \ref{sec:evmetric}), the standard deviation of the latter and the runtimes. In order to make the comparison as fair as possible, all the estimators are implemented with the rational kernel. We found out that the performances drop with the more standard Gaussian kernel (which does not apply to RVDE). We include the results with both the Gaussian kernel and the exponential kernel in the Appendix.

\begin{center}
\begin{table}[!tbh]
\setlength\extrarowheight{4pt}
\centering
\caption{Average runtimes (in seconds) per one full train-test run with fixed bandwidth. RVDE is highlighted in blue.  }
\label{tab:runtimes}
\vspace{1em}
\begin{tabular}{rcccccc}
\toprule 
 & RVDE  & CVDE & KDE & AdaKDE\\ 
\hline
\hline

Gaussian & 0.0376 & 0.265 & 0.0340 & 0.266 \\
Anuran Calls  & 0.0581 & 0.490 & 0.0787 & 0.870 \\
MNIST         & 17.4 & 408 & 12.5 & 75.0 \\
\bottomrule
\end{tabular}

\begin{tikzpicture}[overlay]
\draw[fill=RoyalBlue, opacity=.15, draw=none] (-1.8,-.2) rectangle (-.4,2.8);
\end{tikzpicture}
\end{table}
\end{center}

The plot in Figure \ref{fig:de-comparison} displays the (test) log-likelihood for all the estimators as the bandwidth hyperparameter $h$ varies (see the definition of the baselines in Section \ref{sec:evmetric}). In order to compare RVDE on the same scale as the other estimators, we convert $h$ to $\alpha$ via:
\begin{equation}
\alpha = \int_0^\infty K\left( \frac{t}{h} \right) \ \textnormal{d}t = h^n \int_0^\infty K(t) \ \textnormal{d}t.
\end{equation}
As can be seen, RVDE outperforms the baselines (each with its respective best bandwidth) in all the cases considered. The margin between RVDE and the baselines is especially evident on the more complex and high-dimensional datasets (Anuran Calls and MNIST). This confirms that the geometric benefits and the continuity properties of RVDE translate into better estimates for densities of different nature and increasing dimensionality.

Table \ref{tab:runtimes} reports the average runtime for an experimental run (with a single fixed bandwidth) for each estimator. RVDE outperforms the CVDE as well as AdaKDE by an extremely large margin. KDE achieves comparable runtimes to RVDE: it is slightly faster on Gaussian and MNIST while it is slightly slower on Anuran Calls. This confirms empirically the discussion from Section \ref{sec:compute}: RVDE is significantly more efficient than CVDE and has the same (asymptotic) complexity as KDE.  

% Note that for MNIST, RVDE is twice as fast as KDE. This is probably due to the fact that KDE evaluates the kernel at each datapoint. RVDE instead computes $K$ once in Equation \ref{contvde} and multiple times for each iteration of the numerical solver approximating $\beta$. The latter is independent on $|P|$ and can result in fewer evaluations for larger datasets such as MNIST. 

\begin{center}
\begin{table}[!tbh]
\setlength\extrarowheight{4pt}
\centering
\caption{Standard deviations of the (test) log-likelihood over $5$ experimental runs. RVDE is highlighted in blue. Each estimator is considered with its best bandwidth.}
\label{tab:stds}
\vspace{1em}
\begin{tabular}{rcccccc}
\toprule 
 & RVDE  & CVDE & KDE & AdaKDE\\ 
\hline
\hline

Gaussian & 0.788 & 0.843 & 0.572 & 0.572 \\
Anuran Calls  & 1.170 & 1.253 & 1.152 & 1.152 \\
MNIST & 5.507 & 5.767  & 5.735 & 5.735 \\
% MNIST $\left(\times 10^{-18}\right)$  & 1.127  &  & 1.331  & 1.292  \\
\bottomrule
\end{tabular}

\begin{tikzpicture}[overlay]
\draw[fill=RoyalBlue, opacity=.15, draw=none] (-1.7,-.2) rectangle (-.3,2.8);
\end{tikzpicture}
\end{table}
\end{center}

Table \ref{tab:stds} separately reports the standard deviation of the log-likelihood (averaged over $P_\textnormal{test}$) w.r.t. the dataset sampling. For each estimator, we consider its best bandwidth according to the results from Figure \ref{fig:de-comparison}. We first observe that RVDE achieves lower standard deviation than CVDE on all the datasets. This corroborates the hypothesis that the continuity of RVDE results in more stable estimates than those obtained by the highly-discontinuous CVDE. KDE and AdaKDE achieve the lowest standard deviations on the Gaussian and Anuran Calls datasets. This is likely due to the smoothness of such estimators and again confirms the benefit of regularity biases in terms of stability. However, on the most complex dataset considered (MNIST) RVDE outperforms the baselines. This suggests that for articulated densities the biases of geometric nature become more beneficial than generic biases such as smoothness.

\section{CONCLUSIONS AND FUTURE WORK}
In this work we introduced a non-parametric density estimator (RVDE) benefiting from the geometric properties of Voronoi tessellations while being continuous and computationally efficient. We provided both theoretical and empirical investigations of RVDE. 

An interesting line for future investigation is to explore the radial construction of RVDE on Riemannian manifolds beyond the Euclidean space. In this generality, the rays correspond to geodesics defined via the exponential map of the given manifold. A variety of Riemannian manifolds arise in statistics and machine learning. For example, data on spheres are the object of study of directional statistics (\cite{directional}), hyperbolic spaces are routinely deployed to represent hierarchical data (\cite{nickel2017poincare}) and complex projective spaces correspond to Kendall shape spaces from computer vision (\cite{klingenberg2020walking}). Those areas of research can potentially benefit from the geometric characteristics and the computational efficiency of an extension of RVDE to Riemannian manifolds.

\section*{Acknowledgements}
This work was supported by the Swedish Research Council, Knut and Alice
Wallenberg Foundation and the European Research Council (ERC-BIRD-884807).

\printbibliography

%%%%%%%%%%%%%%%%%%%%%%%%%%%%%%%%%%%
%%%%%% SUPPLEMENT (OPTIONAL) %%%%%%
%%%%%%%%%%%%%%%%%%%%%%%%%%%%%%%%%%%
% \clearpage
% \thispagestyle{empty}

\onecolumn 
\aistatstitle{An Efficient and Continuous Voronoi Density Estimator: \\ Supplementary Materials}

\section*{PROOFS OF RESULTS FROM SECTION \ref{sec:study}}
\begin{prop}
Fix $l>0$ and suppose $K \in C^1(\mathbb{R}_{>A})$. Then the iteration $\beta_{m+1}$ of the Newton-Raphson method for computing $\beta(l)$ by solving Equation \ref{eqbeta} takes form: 
\begin{equation}\label{eqnewtonappend}
\beta_{m+1} = \beta_m\left( 1 + \frac{1}{n} \left(  1 - \frac{l^nK(\beta_m l) - n\alpha}{l^nK(\beta_m l)  - n \int_0^l t^{n-1}K(\beta_mt) \ \textnormal{d}t}  \right) \right). 
\end{equation}
Moreover, if $K$ is convex then the Newton-Raphson method converges for any initial value $\beta_0$ i.e., $\lim_{m \to \infty}\beta_m = \beta(l)$. 
\end{prop}

\begin{proof}
Consider 
\begin{equation}
    F(\beta) = \int_0^l t^{n-1}K(\beta t) \ \textnormal{d}t - \alpha.
\end{equation}
The iteration of the Newton-Rhapson method for solving $F(\beta) = 0$ takes form: 
\begin{equation}\label{eqappendraw}
    \beta_{m+1}  = \beta_m - \frac{F(\beta_m)}{\frac{\textnormal{d}F}{\textnormal{d}\beta}(\beta_m)}.
\end{equation}
Via integration by parts we obtain: 
\begin{equation}
\frac{\textnormal{d}F}{\textnormal{d}\beta}(\beta) = \int_0^l t^n \frac{\textnormal{d}K}{\textnormal{d}t}(\beta t) \ \textnormal{d}t = \frac{1}{\beta} \left(  l^n K(\beta l) -n  \int_0^l t^{n-1}K(\beta t) \ \textnormal{d}t  \right). 
\end{equation}
Equation \ref{eqnewtonappend} follows then from Equation \ref{eqappendraw} by elementary algebraic manipulations. The convergence guarantee follows from the fact that if $K$ is convex then $F$ is easily seen to be convex as well. The Newton-Raphson method is well-known to be convergent for convex functions (\cite{boyd2004convex}).   
\end{proof}

\begin{prop}\label{diffeqappendix}
The function $\beta: \ \mathbb{R}_{>0} \rightarrow \mathbb{R}$ is increasing, has a zero at $l= (n \alpha)^{\frac{1}{n}}$ and has an horizontal asymptote: 
\begin{equation}
\lim_{l \to + \infty} \beta(l) = \left( \frac{1}{\alpha} \int_0^\infty t^{n-1}K(t) \ \textnormal{d}t \right)^{\frac{1}{n}}. 
\end{equation}
Moreover if $K \in C^1(\mathbb{R}_{>A})$ then $\beta \in C^1(\mathbb{R}_{>0})$ and it satisfies the differential equation:
\begin{equation}
 \left( l- \frac{n \alpha}{l^{n-1}K(\beta(l)l)} \right) \frac{\textnormal{d}\beta}{\textnormal{d}l}(l) = - \beta(l). 
\end{equation}
\end{prop}

\begin{proof}
 The claim on the monotonicity of $\beta$ follows directly from its definition (Equation \ref{eqbeta}) and the hypothesis that $K$ is decreasing. In order to compute its zero, note that $\beta(l)=0$ implies $\alpha = \int_0^l K(0)t^{n-1} = \frac{l^n}{n}$ and thus $l = (n \alpha)^{\frac{1}{n}}$. For the asymptote note that for $l = + \infty$ Equation \ref{eqbeta} becomes  by a change of variables: 
\begin{equation}
\int_0^\infty t^{n-1}K(\beta(+\infty)t) \ \textnormal{d}t = \frac{1}{\beta(+\infty)^n} \int_0^\infty t^{n-1}K(t) \ \textnormal{d}t = \alpha.
\end{equation}

 Lastly, in order to obtain the differential equation for $\beta$ we differentiate Equation \ref{eqbeta} on both sides and get: 
  \newpage 

 \begin{equation}
  \begin{split}
 0 = \frac{\textnormal{d}}{\textnormal{d}l} \int_0^l t^{n-1}K(\beta(l)t) \ \textnormal{d}t & = l^{n-1}K(\beta(l)l) + \int_0^l t^{n-1}\frac{\textnormal{d}}{\textnormal{d}l}  K(\beta(l)t) \ \textnormal{d}t  \\
 &=l^{n-1}K(\beta(l)l) + \frac{\textnormal{d}\beta}{\textnormal{d}l}(l) \int_0^l t^{n}\frac{\textnormal{d}K}{\textnormal{d}t} (\beta(l)t) \ \textnormal{d}t  \\
 & = l^{n-1}K(\beta(l)l) + \frac{\textnormal{d}\beta}{\textnormal{d}l}(l) \frac{l^{n}K(\beta(l)l) - n\alpha}{\beta(l)} 
 \end{split}
 \end{equation}
 where in the first identity we deployed the (distributional) Leibniz rule while in the last one we deployed integration by parts. 

\end{proof}

\begin{prop}
The modes of $f_P$ are as follows:  
\begin{enumerate}[label=(\arabic*)]
    \item $p \in P$ if $d(p, q)> 2 \varepsilon$ for every Voronoi cell $C(q)$ adjacent to $C(p)$,
    \item $\frac{p + q}{2}$ for $p, q \in P$ if $\frac{p + q}{2} \in C(p) \cap C(q)$ and $d(p,q) < 2 \varepsilon$,
    \item all points belonging to the segment $[p,q]$ for $p, q \in P$ if $\frac{p + q}{2} \in C(p) \cap C(q)$ and $d(p,q) = 2\varepsilon$. 
\end{enumerate}
\end{prop}

\begin{proof}
 Pick $p\in P$. If $p$ satisfies the hypothesis of the first claim then $l(x) > \varepsilon$ for every $x\in C(p)$ and thus $\beta(l(x)) > 0$ by Proposition \ref{diffeqappendix}. Since $K$ is decreasing, $f_P$ decreases radially w.r.t. $p$ in $C(p)$ and the first claim follows. If $p$ does not satisfy the hypothesis of the first claim then $\beta(l(x)) \leq 0$ for some $x \in C(p)$. With the exception of the case $\beta(l)=0$, the modes lie then on the boundary and are of the form $K(\beta(l)l)$ up to a multiplicative constant. The function $\beta(l)l$ is increasing in $l$ since  by appealing to Proposition \ref{diffeqappendix} we can compute its derivative: 
 \begin{equation}
  \frac{\textnormal{d} \beta(l)l}{\textnormal{d}l} = \beta(l) + l   \frac{\textnormal{d} \beta(l)}{\textnormal{d}l}  = \beta(l) \frac{n\alpha}{n\alpha - l^n K(\beta(l)l)} \geq 0.
 \end{equation}

Since $l(x)$ has local minima at midpoints of segments connecting points in $P$, $K(\beta(l)l)$ is locally maximized therein and the second claim follows. In the hypothesis of the third claim $\beta$ vanishes on the segment and the density is thus constant. 

\end{proof}

% \newpage
\section*{ADDITIONAL EXPERIMENTS}
In this section we report additional experimental results complementing the ones in the main of the manuscript. For completeness, we evaluate the density estimators on different kernels. Figure \ref{fig:de-comparison-e} displays comparative results for all the estimators with the exponential and the Gaussian kernel (note that the latter does not apply to RVDE). Moreover, we experiment with different dimensions and evaluation metrics other than average log-likelihood. This is possible only on a synthetic dataset where the dimension can vary and where the ground-truth density $\rho$ is known. The latter is necessary for the metric considered. Figure \ref{fig:hellinger} displays a comparison on a high-dimensional Gaussian mixture ($n=30$) as well as a comparison on the Gaussian mixture as in Section \ref{sec:exper} ($n=10$) where the evaluation metric is the empirical \emph{Hellinger distance} on the test set:

\begin{equation}
\frac{1}{2 |P_{\rm test}|} \sum_{x \in P_{\rm test}}\left( f_P(x)^{\frac{1}{2}} - {\rho}(x)^{\frac{1}{2}} \right)^2.
\end{equation}

\begin{figure}[!th]
    \centering
    \includegraphics[width=.6\linewidth]{images/de_comparison/vskde_labels.png}
    
    \begin{subfigure}{.225\linewidth}
        \centering
        \includegraphics[width=\linewidth]{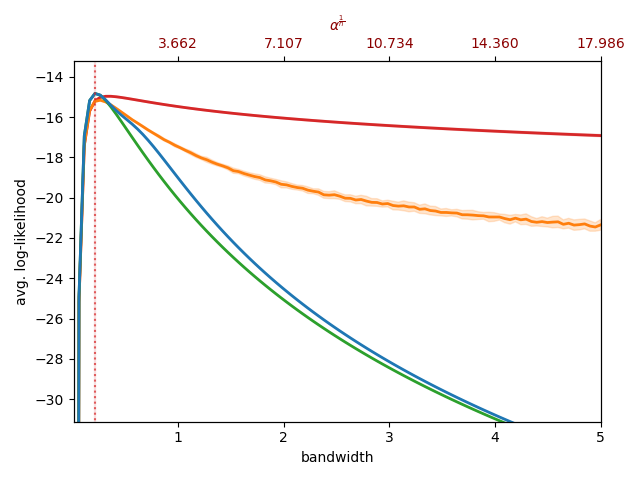}
        % \subcaption*{Single Gaussian}
    \end{subfigure}
    \begin{subfigure}{.225\linewidth}
        \centering
        \includegraphics[width=\linewidth]{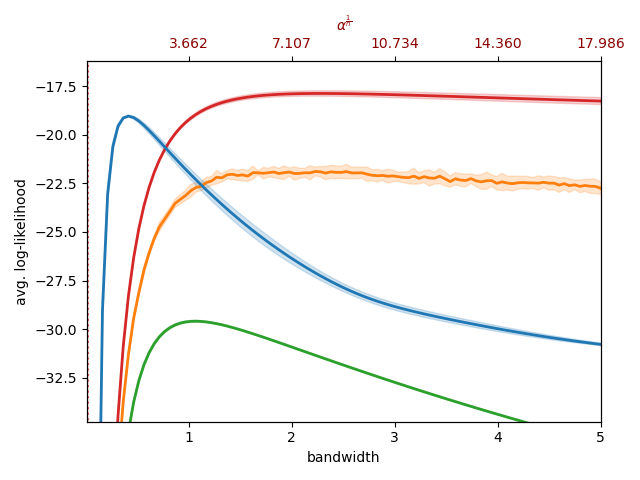}
        % \subcaption*{Gaussian Mixture}
    \end{subfigure}
    \begin{subfigure}{.225\linewidth}
        \centering
        \includegraphics[width=\linewidth]{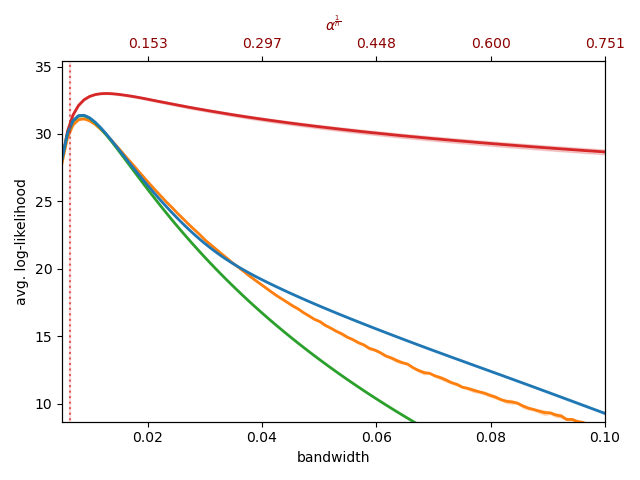}
        % \subcaption*{Anuran Calls}
    \end{subfigure}
    \begin{subfigure}{.225\linewidth}
        \centering
        \includegraphics[width=\linewidth]{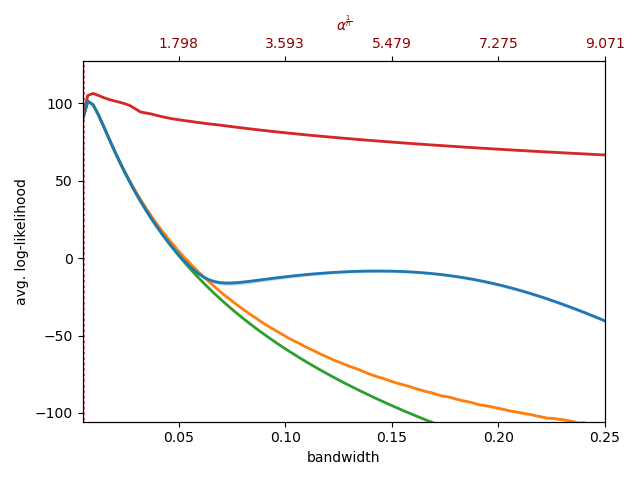}
        % \subcaption*{MNIST}
    \end{subfigure}
    % \caption{Comparison of the estimators as the bandwidth varies, exponential kernel $K(t) = e^{-t}$.}
%     \label{fig:de-comparison-e}
% \end{figure}

% \begin{figure}[!tbh]
%     \centering
    \begin{subfigure}{.225\linewidth}
        \centering
        \includegraphics[width=\linewidth]{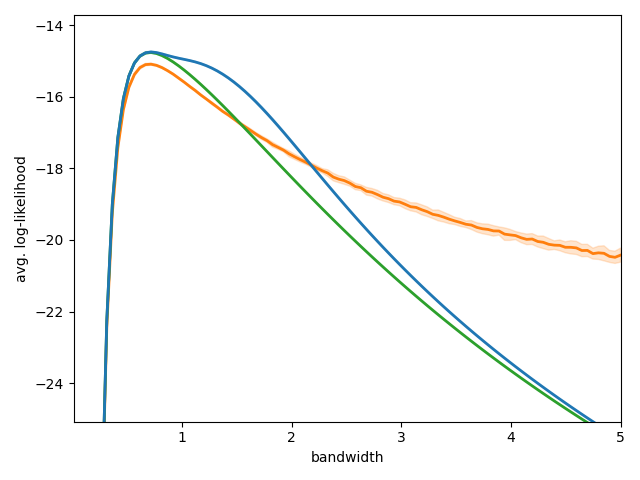}
        \subcaption*{Single Gaussian}
    \end{subfigure}
    \begin{subfigure}{.225\linewidth}
        \centering
        \includegraphics[width=\linewidth]{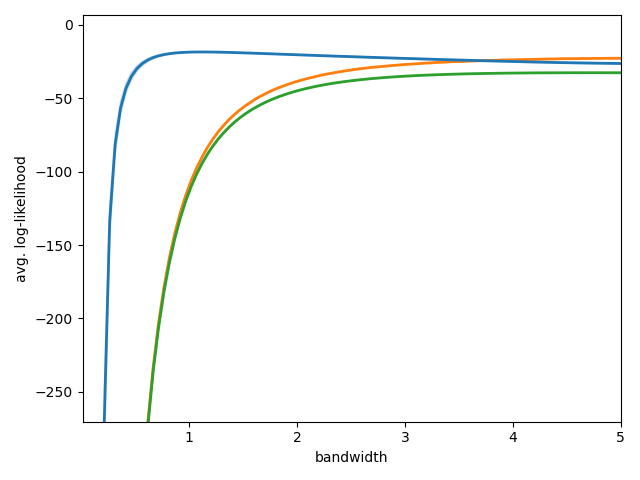}
        \subcaption*{Gaussian Mixture}
    \end{subfigure}
    \begin{subfigure}{.225\linewidth}
        \centering
        \includegraphics[width=\linewidth]{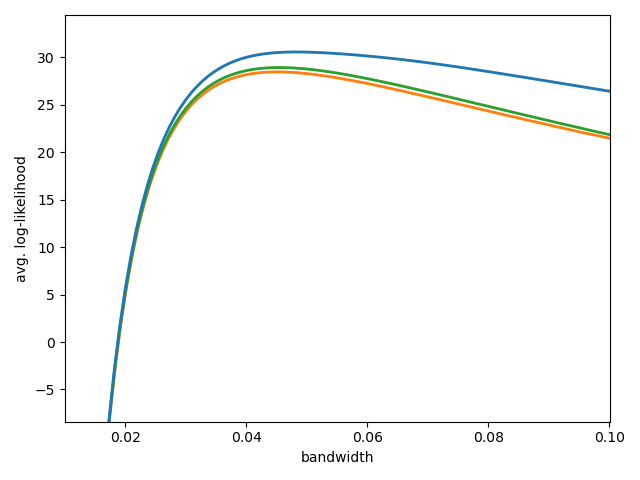}
        \subcaption*{Anuran Calls}
    \end{subfigure}
    \begin{subfigure}{.225\linewidth}
        \centering
        \includegraphics[width=\linewidth]{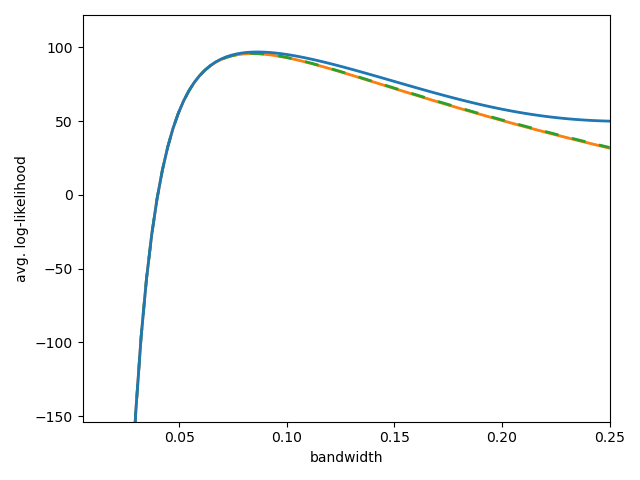}
        \subcaption*{MNIST}
    \end{subfigure}
    \begin{picture}(0,0)
    \put(-463,72){\rotatebox{90}{\small Exponential}}
    \put(-463,-3){\rotatebox{90}{\small Gaussian}}
\end{picture}
    \caption{Comparison of the estimators with the exponential and Gaussian kernel as the bandwidth varies.}
    \label{fig:de-comparison-e}

\vspace{5em}
    \begin{subfigure}{.3\linewidth}
    \centering
        \includegraphics[width=\linewidth]{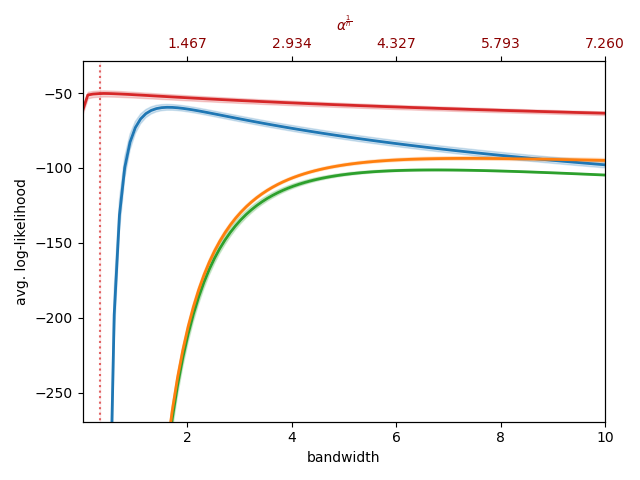}
        \subcaption*{$n=30$}
    \end{subfigure}
    \hspace{3em}
    \begin{subfigure}{.3\linewidth}
    \centering
        \includegraphics[width=\linewidth]{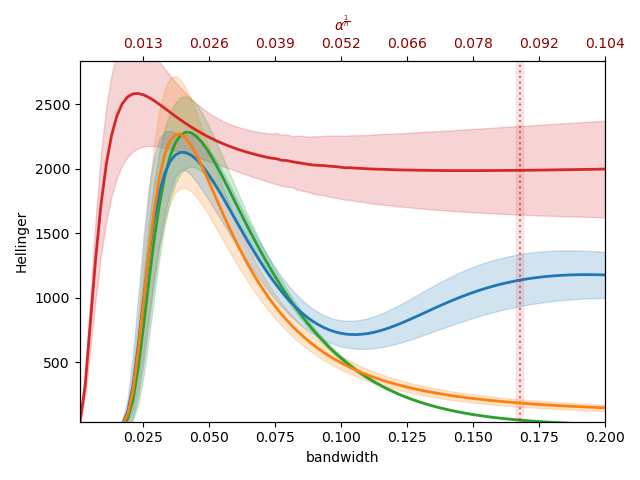}
        \subcaption*{Hellinger Distance}
    \end{subfigure}
    \caption{Comparison of the estimators on a $30$-dimensional Gaussian mixture (left) and on a $10$-dimensional Gaussian mixture with the Hellinger distance as a metric (right).}
    \label{fig:hellinger}
    \vspace{15em}
\end{figure}

\end{document}